\documentclass[journal,10pt,twocolumn]{IEEEtran}
\usepackage[colorlinks,bookmarksopen,bookmarksnumbered,citecolor=red,urlcolor=red,]{hyperref}

\usepackage{color}
\usepackage{amsmath,bm}

\usepackage{amsmath,mathtools}

\usepackage{verbatim}
\usepackage{epsfig,bbm}
\usepackage[colorlinks,bookmarksopen,bookmarksnumbered,citecolor=red,urlcolor=red,]{hyperref}
\usepackage{CJK}
\usepackage{indentfirst}
\usepackage{multirow}
\usepackage{epstopdf}
\usepackage{graphicx}
\usepackage{footmisc}
\graphicspath{{./figures/}}
\usepackage{amsfonts}
\usepackage{mathrsfs}
\usepackage{setspace}
\usepackage{amsmath}
\usepackage{algorithm,algorithmic,amsbsy,amsmath,amssymb,epsfig,bbm,mathrsfs, bbm} 
\usepackage{amsthm}
\usepackage{verbatim}
\usepackage[noadjust]{cite}
\usepackage[subfigure]{graphfig}

\newtheorem{theorem}{Theorem}

\newtheorem{corollary}{Corollary}

\usepackage{geometry}
\usepackage{xcolor}
\allowdisplaybreaks
\usepackage{xcolor}
\allowdisplaybreaks

\begin{document}
 
\title{On Coverage Probability With Type-II HARQ in Large Uplink Cellular Networks \vspace{2mm}}

 \author{ Xiao Lu, Ekram Hossain, Hai Jiang, and Guangxia Li \vspace{-5mm} \thanks{X. Lu, and H. Jiang are with the Department of ECE, University of Alberta, Canada (emails: hai1@ualberta.ca). E. Hossain is with the Department of ECE, University of Manitoba (email: Ekram.Hossain@umanitoba.ca). Guangxia Li is with the School of Computer Science and Technology, Xidian University, China (email:gxli@xidian.edu.cn). This work was done when the first author visited the University of Manitoba, Canada}}  
\maketitle

\begin{abstract} 
This letter studies uplink transmission in large-scale cellular networks with a Type-II hybrid automatic repeat request (HARQ) retransmission scheme, under which an unsuccessful transmission (if occurs) is combined with the corresponding retransmission through maximum-ratio combining (MRC) for decoding. Based on stochastic geometry analysis, the uplink coverage probabilities are characterized under a generalized power control scheme in the scenarios with quasi-static interference (QSI) and fast-varying interference (FVI), where the same or different interfering users are present during the transmission and retransmission phase, respectively. Our analytical expressions reveal some scaling properties of the coverage probabilities and can be used to evaluate the diversity gain of MRC (i.e., the ratio of the required signal-to-interference ratio (SIR) to achieve a target coverage probability with MRC to that without MRC).  We show that the diversity gain of MRC is more remarkable in the scenario with QSI compared to that with FVI. Moreover, the diversity gain of MRC can be mostly exploited by adopting full channel-inversion power control.  

\end{abstract}

 
\section{Introduction}

Wireless communications are subjected to performance degradation caused by fading and interference fluctuations. Retransmission diversity is an essential approach used in
wireless systems to combat such performance degradation and provide reliable data transfer. 
Hybrid automatic repeat request (HARQ) is a key retransmission technique adopted
in many existing standards, like 3GPP-LTE and WiMAX. Type-I HARQ and Type-II HARQ with chase combining (HARQ-CC)~\cite{G.2015Nigam} are two major HARQ techniques. The former decodes the received signal of each transmission independently while the latter combines the received signals of different transmissions with maximum-ratio combining (MRC) for decoding.

Analysis of retransmission schemes in large-scale networks is fundamental to the understanding of retransmission diversity in practical communication systems.
Only a few prior works have investigated HARQ schemes in large-scale cellular networks.  
Retransmission diversity loss in static networks is a key observation reported in the existing literature~\cite{G.2015Nigam,M.Nov.2013Haenggi,H.Dec.2016Afify}. The diversity loss is due to temporal interference correlation because the initial transmission and the corresponding retransmission are affected by the same interferers. 
References~\cite{G.2015Nigam} and \cite{M.2017Sheng} study both Type-I and Type-II HARQ-CC in downlink heterogeneous networks with and without base station (BS) cooperation, respectively. The focus of~\cite{H.Dec.2016Afify} is a unified analysis of Type-I HARQ in downlink cellular networks with different multiple-antenna configurations.  
These works highlight the presence of retransmission diversity loss, however, do not quantify its effect. 
The authors in~\cite{J.July2018Wen} 
characterize the interference correlation in terms of Pearson's correlation coefficient~\cite{M.2013Haenggi} in a cluster Poisson network with Type-I HARQ and find that interferer clustering increases the interference correlation. Different from the above literature that target downlink transmissions without power control, \cite{R.2018Arshad} studies Type-I HARQ in uplink cellular networks under fractional power control (FPC). To the best of our knowledge, none of the existing literature has investigated Type-II HARQ-CC in uplink cellular networks.

In this letter, we analyze the
effects of power control on the transmission diversity loss and the diversity gain of MRC under Type-II HARQ in large-scale uplink cellular networks.  
We consider two practical types of interference, namely, the quasi-static interference (QSI) and the fast-varying interference (FVI), which are suitable for modeling 
static/low-mobility and high-mobility scenarios, respectively~\cite{Z.Feb.2009Win}.  
Specifically, in the former and latter scenarios, the initial transmission and retransmission (if there is) are affected by the same and different set of the interferers, respectively. As a result, the interference imposed on transmission and retransmission are temporarily correlated and independent, respectively. Our developed analytical framework reveals some scaling properties and provides insights into the effects of network parameters (e.g. power control parameters) on the uplink coverage performance. 

 The rest of the letter is organized as follows. Section II presents the system model and assumptions. The analytical results on the uplink coverage probability with Type-II HARQ-CC are presented in Section III. Section IV presents the numerical results and Section V concludes the letter.
 
\vspace{-1mm}
\section{System Model and Assumptions}

We consider an uplink Poisson cellular network, where the BSs $\Phi_{B}$ are distributed as a homogeneous Poisson point process (PPP) with spatial density $\zeta_{B}$ and each BS is associated with a uniformly distributed user on each universally reused resource block, referred to as the user model of type I in~\cite{M.April2017Haenggi}. Note that  
due to the correlation of neighboring Voronoi cells of the BSs, the resulted user point process $\Phi_{U}$ for each resource block is not a PPP.
Without loss of generality, we analyze the uplink transmission performance of a typical user, denoted as $0$. The serving BS of the typical user, named tagged BS, is located at the origin.  
We employ the power-law path-loss model with path-loss with exponent $\alpha >2$ and small-scale fading. The fading gain, denoted as $h_{i}$, between a user $i \in \Phi_{U}$ and its serving BS, is assumed to be an independent and identically distributed (i.i.d.) exponential random variable with unit mean.  
Moreover, $\mathbf{h}=\{h_{i}| i \in \Phi_{U} \}$ are assumed to be i.i.d. and vary across different transmission attempts.

A generalized fractional power control (GFPC) scheme~\cite{D2016Renzo} is adopted for uplink transmission, which sets the transmit power of any user $i \in \Phi_{U}$ as  
\begin{equation}\label{PC}  
	P_{i} =  
	\begin{cases}
\varrho l^{\alpha\epsilon}_{i}   &   \text{ if }  \varrho l^{\alpha\epsilon}_{i}  \leq  \widehat{P}  ,\\
	\bar{P} & \text{ otherwise} ,
	\end{cases}   
\end{equation}
where $\varrho$ is the baseline transmit power,  $l_{i}$ denotes the link distance between user $i$ and its associated BS, $\epsilon \in [0,1]$ is the path-loss compensation exponent (PCE), $\widehat{P}$ and $\bar{P}$, respectively, denotes the maximum transmit power and the enforced transmit power if $P_{i}$ exceeds $\widehat{P}$.
Note that (\ref{PC}) generalizes several power control schemes of interest. Specifically, the GFPC is equivalent to i) FPC  \cite{Y.2017Wang} when $\widehat{P}  = \infty$; ii) truncated fractional power control (TFPC) \cite{H.2014ElSawy} when $\widehat{P} < \infty$ and $\bar{P}=0$; iii) full channel inversion power control (FCIPC) when $\epsilon=1$ and $\widehat{P}=\infty$ \cite{Y.2017Wang}; and iv) no power control (NPC) when $\epsilon=0$.

The cellular network uses a Type-II HARQ-CC as the retransmission scheme.  
Specifically, the tagged BS requests the retransmission, denoted as $\mathrm{R}$, if the receive signal-to-interference ratio (SIR) of the initial transmission, denoted as $\mathrm{T}$, is below a pre-defined threshold $\tau$. Upon receiving the retransmission, the tagged BS performs decoding from the combined signals of the two transmissions based on MRC. 
With the GFPC, the uplink SIR at the tagged BS either 
for an initial transmission or the corresponding retransmission (if occurs) is calculated as 
\begin{align} 
\hspace{-2mm} \eta^{(t)} \!=\! \frac{ P_{0} h_{0} d^{-\alpha}_{0} }{ \sum_{j \in \Phi_{\!I} } \! \! P_{j} h_{j} d^{-\alpha}_{j} } \! \overset{(a)}{=} \! \frac{  h_{0}  l^{\alpha(1-\epsilon)}_{0}     }{ \sum_{j \in \Phi_{\!I} } \! l^{\alpha \epsilon }_{j}  h_{j}  d^{-\alpha}_{j} }, \, t \!\in\! \{\mathrm{T}, \mathrm{R}\},  \vspace{-5mm}
\end{align} 
where $d_{j}$ denotes the distance between user $j$ and the tagged BS, $\Phi_{I}:=\Phi_{U} \backslash \{0\}$, and $(a)$ follows as $l_{0}=d_{0}$.

Let $ \!\Phi^{(t)}_{\!B}, \Phi^{(t)}_{u} \!$ represent the realizations of the BS and the user point processes, respectively, during a transmission attempt $t \in \{ \mathrm{T},  \mathrm{R} \}$ of the typical user.
Note that $\Phi^{(\mathrm{T})}_{\!B}$ ($\Phi^{(\mathrm{T})}_{\!u}$) is identical to and different from $\Phi^{(\mathrm{R})}_{\!B}$ ($\Phi^{(\mathrm{R})}_{\!u}$), in the scenarios with QSI and FVI, denoted as $\mathrm{Q}$ and $\mathrm{F}$, respectively. In the following, the superscript $(t)$ is dropped for the scenario with QSI and kept for the scenario with FVI to indicate the identity and difference of the network realizations, respectively.   
Given the type of interference experienced, i.e. either $\mathrm{Q}$ or $\mathrm{F}$, the coverage probability  
is defined as
\begin{align} \label{def}
  \mathbf{C}^{\mathrm{A}} 
  \!\! := \! \underbrace{ 
 \mathbb{P} \big[ \eta^{(\mathrm{T})} \!>  \!  \tau 
 \big] 
  }_{:=\mathbf{C}^{(\mathrm{T})}}   
 \! +     
 \mathbb{P} \big[ \eta^{(\mathrm{R})} \! > \!  \tau \! - \! \eta^{(\mathrm{T})}\! , \eta^{(\mathrm{T})}  \! \leq  \!   \tau | \Xi^{\mathrm{A}} \big],  
\end{align} 
where $\!\mathrm{A} \! \in \! \{ \mathrm{Q},\mathrm{F} \}\!$ denotes the interference type indicator, $\mathbb{P}[Z]$ represents the probability that the event $Z$ happens, $\mathbf{C}^{(\mathrm{T})} $ represents the coverage probability of the initial transmission,  $\Xi^{\mathrm{Q}} :=\{  \Phi_{\!B}, \Phi_{\!u} \}$ and  $\Xi^{\mathrm{F}} \! : = \! \big \{   \Phi^{(\mathrm{T})}_{\!B} , \Phi^{(\mathrm{T})}_{\!u}, \Phi^{(\mathrm{R})}_{\!B} , \Phi^{(\mathrm{R})}_{\!u} \big \}$ are the sets of network realizations in the scenarios with QSI and FVI, respectively.  
 
\section{Analysis of Uplink Coverage Probability}

This section characterizes the expectation of the coverage probability of a typical uplink user 
considering different settings of the power control parameters for scenarios with QSI and FVI. We start with the general results as follows. 

\begin{theorem} \label{Thm1}  
The uplink coverage probability with GFPC can be closely approximated by (\ref{eqn:M_II}) shown on the top of the next page, 
\begin{figure*}
\begin{align} \label{eqn:M_II}    
& \hspace{-3mm} \widetilde{\mathbf{C}}^{\mathrm{A}}_{\mathrm{GFPC}} \!\! =  \! \! 
\begin{dcases}   
 & \hspace{-4mm} \widetilde{\mathbf{C}}^{(\mathrm{T})}_{\mathrm{GFPC}} \! + \! \! \int^{\tau}_{0} \!\!\! \Bigg( \!\! \int^{B_{\!\widehat{P}} }_{0} \!\!\! \xi_{\tau,\eta,t} \Big(  \frac{u^{ \! \frac{\alpha \epsilon }{2}    } }{  v^{ \! \frac{\alpha}{2}  } }, \frac{B^{ \! \frac{\alpha \epsilon}{2}   } _{\!\bar{P}} }{t^{ \! \frac{\alpha  \epsilon}{2}  } \! v^{\frac{\alpha}{2} }}  \Big)  \mathrm{d}t    + \int^{\infty}_{B_{\!\widehat{P}} } \!\!  
\xi_{\tau,\eta,t} \Big(  \frac{  u^{\!\frac{\alpha \epsilon}{2}}t^{\frac{\alpha \epsilon}{2}  }  }{B^{\frac{\alpha \epsilon}{2}  } _{\!\bar{P}} \! v^{\!\frac{\alpha}{2} }  },v^{-\!\frac{\alpha}{2} }  \Big) 
\mathrm{d}  t \! \Bigg) \mathrm{d}  z, \hspace{35mm} \mathrm{A} = \mathrm{Q}   \\
& \hspace{-4mm} \widetilde{\mathbf{C}}^{(\mathrm{T})}_{\mathrm{GFPC}} \! + \! \! \int^{\tau}_{0} \!\!\! \Bigg( \!\!   \int^{B_{\!\widehat{P}}}_{0} \!\!\! \exp \! \bigg( \!  \! - \! t     -     \omega_{\tau-\eta,0,t} \Big( \frac{u^{ \! \frac{\alpha \epsilon }{2}    } }{  v^{ \! \frac{\alpha}{2}  } }\!, \! \frac{B^{ \! \frac{\alpha \epsilon}{2}   } _{\!\bar{P}} }{t^{ \! \frac{\alpha  \epsilon}{2}  } \! v^{\frac{\alpha}{2} }} \! \Big)   \!   \bigg) \mathrm{d}t \!    \int^{B_{\!\widehat{P}} }_{0} \!\!\!\!  \xi_{\eta,\eta,t} \Big(\frac{u^{ \! \frac{\alpha\epsilon}{2} }}{v^{ \! \frac{\alpha}{2} }}, \frac{B^{ \frac{\alpha \epsilon}{2}   }_{\!\bar{P}}}{  t^{ \! \frac{\alpha \epsilon}{2}    } \! v^{\!\frac{\alpha}{2} }} \! \Big)  
\mathrm{d}t  \!  + \! \!  \int^{\infty}_{\!B_{\!\widehat{P}} } \!\!  
\xi_{\eta,\eta,t} \Big(  v^{-\!\frac{\alpha}{2} } \!, \!  \frac{  u^{\!\frac{\alpha \epsilon}{2}}t^{\frac{\alpha \epsilon}{2}  }  }{B^{\frac{\alpha \epsilon}{2}  } _{\!\bar{P}} \! v^{\!\frac{\alpha}{2} }  }\Big) 
\mathrm{d}  t  \\
& \hspace{70mm} \times \!  \int^{\infty}_{\!B_{\!\widehat{P}} } \! \exp \! \bigg(  \! \! - \! t - \!    \omega_{\eta,\eta,t} \Big ( \frac{  u^{\!\frac{\alpha \epsilon}{2}} \! t^{\frac{\alpha \epsilon}{2}  }  }{B^{\!\frac{\alpha \epsilon}{2}  } _{\!\bar{P}} \! v^{\!\frac{\alpha}{2} }  }  ,   v^{-\frac{\alpha}{2} }   \! \Big)   \!  \bigg) \mathrm{d} t \Bigg) \mathrm{d}  z,  \hspace{2mm} \mathrm{A} = \mathrm{F} 
\end{dcases}\hspace{-3mm}
\end{align} 
\hrulefill
\end{figure*}
where $\widetilde{\mathbf{C}}^{(\mathrm{T})}_{\mathrm{GFPC}} $ is 
\begin{align}  \label{eqn:S}   
 & \widetilde{\mathbf{C}}^{(\mathrm{T})}_{\mathrm{GFPC}}  \! =   \!  
      \int^{ B_{\!\widehat{P}} }_{0} \!\!\!  \exp \! \bigg(\! \! - \! t     -      \omega_{\tau,0,t}   \Big( 
  \frac{u^{\frac{\alpha \epsilon}{2}  }}{ v^{ \frac{\alpha}{2} }} , 
  \frac{B^{\frac{\alpha\epsilon}{2} }_{\!\bar{P}}  }{v^{\frac{\alpha}{2}} t^{\frac{\alpha \epsilon}{2} }} \!
  \Big)      \!  \bigg) \mathrm{d}t \nonumber \\
  & \hspace{15mm} + \int^{\infty}_{\!B_{\!\widehat{P}}} \!  \exp \! \bigg( \!\! - \!   t     -    \omega_{\tau,0,t} \Big(   \frac{u^{\!\frac{\alpha\epsilon}{2} } t^{\frac{\alpha\epsilon}{2} } }{ B^{\frac{\alpha\epsilon}{2}}_{\!\bar{P} } \!  v^{\frac{\alpha}{2}}  } \!, 
     v^{-\frac{\alpha}{2}}  \!
    \Big)  \! \Big]  \! \bigg) \mathrm{d}t    ,     
\end{align}  
and $\!B_{a} \! = \! C_{2} \pi \zeta_{B} (\frac{a}{\varrho})^{\frac{2}{\alpha \epsilon}}  $,  $C_{1}\!=\!\frac{12}{5}$, $C_{2}\!=\!\frac{5}{4}$,  $\omega$ and $\xi$ are given, respectively, as  
\begin{align} 
&  \hspace{3mm} \omega_{a,b,c}  (x,y) \! = \!   \mho_{0 , c } \Big(\!  Q_{ a,b} \big( 0,    x   \big) \! \Big)\! + \!\,    \mho_{\!\frac{ B_{\!P}}{c} \!, c } \Big( Q_{ a,b}  \big( x, y   \big) \! \Big), \label{omega} \\
& \hspace{-2mm} \text{and}, \,\, \xi_{a,b,c}  (x,y) \! = \! \Big(  \mho_{0, c} ( G_{\!a,b}  (  x ,0     
 ) \big)  \! + \! \mho_{\!\frac{B_{\!P }}{c} \! , c} \big( G_{\! a, b}  (y, x  )   \big) \! \Big)   \nonumber \\
 & \hspace{40mm} \times \exp \! \Big(   \! - t-  
 \omega_{a,b,c} \big( x ,   y \big)   \! 
 \Big), \label{xi}
\end{align}
where $G_{a,b}(x,\!y) \! := \! x f_{a-b}(x) f_{b}(x)^2  \! -    y f_{a-b}(y) f_{b}(y)^2$, $ Q_{a,b}(x,y) \! :=  \! f_{a-b} ( x ) f_{b}  ( x)  \! - \!   f_{a-b} ( y ) f_{b}  ( y)  $, $f_{a}(x)\! := \!\frac{1}{1+ax}$, $   
\mho_{a,b } (x)  \! \!   :=  \! \! \frac{b^2}{C_{2}}    \int^{\infty}_{a} \int^{ v }_{a} \!  x \varepsilon_{b} (u,v) \mathrm{d} u \mathrm{d} v $, therein $ \varepsilon_{b} (u,v)\! :=\! e^{-    ub} \big(1-e^{-\frac{C_{1}}{C_{2}} vb}\big) (1-e^{- v b})^{-1}$.
\end{theorem}

\begin{proof}
The proof is given in \textbf{Appendix A}.
\end{proof}

We notice from (\ref{eqn:M_II}) that in the case when $\widehat{P}= \bar{P} $, $\mathbf{C}^{\mathrm{A}}_{\mathrm{GFPC}}$  
is only affected by the product 
$\zeta_{B} \widehat{P}^{\frac{2}{\alpha \epsilon} } $ for given $\alpha$, $\epsilon$ and $\tau$. This indicates that to guarantee a certain target coverage probability, 
the maximum transmit power of users $\widehat{P}$ can be set inversely proportional to the BS density $\zeta_{B}$ ensuring a fixed value of $\zeta_{B} \widehat{P}^{\frac{2}{\alpha \epsilon} } $. 
 
Next, we investigate the uplink coverage probability in some special cases of the GFPC. We present the analytical results in the following corollaries, the detailed proofs of which are omitted due to the space limit. 
\begin{corollary}\label{C1} 
The uplink coverage probability with FPC  
can be closely approximated by (\ref{C1:FPC}), shown on the top of the next page, where 
$\widetilde{\mathbf{C}}^{(\mathrm{T})}_{\mathrm{FPC}} \!=\! \int^{ \infty }_{0} \!\!  \exp \! \Big (  \! -   t     - \mho_{0, t} \Big( Q_{\tau,0} \big(0,  \!  u^{\!\frac{\alpha \epsilon}{2} } v^{ -\! \frac{\alpha  }{2} }   \big) \! \Big) \! \Big ) \mathrm{d}t $. 
Moreover, in the special case of FCIPC, $\widetilde{\mathbf{C}}^{\mathrm{A}}_{\mathrm{FPC}}$ can be further simplified as (\ref{FCIPC}), where $\!\widetilde{\mathbf{C}}^{(\mathrm{T})}_{ \mathrm{FCIPC}}\!\! =\! \exp \! \Big(  \! -  \mho_{0,1 }  \! \Big (\!  Q_{\tau,0} \big( 0,  
  \! ( \frac{u}{ v })^{\!\frac{\alpha }{2}}  \!  \big)   \!\Big) \!    \Big)$. 
 \begin{figure*} 
 \normalsize  
 \vspace{-6mm}
 \begin{align}  
 \label{C1:FPC}
 \hspace{-0mm} & \widetilde{\mathbf{C}} ^{\mathrm{A}}_{\mathrm{FPC}}    \! \! = \!\!  \nonumber \\
 \hspace{-5mm} & \begin{dcases} \!   & \hspace{-4mm} \widetilde{\mathbf{C}}^{(\mathrm{T})}_{\mathrm{FPC}} \! + \! \! \int^{\tau}_{0} \!\!\!\! \int^{\infty}_{0 } \!  \!      
  \mho_{0,t} \Big( 
  G_{\tau,\eta} \Big(   \frac{u^{\!\frac{\alpha\epsilon }{2} } }{ v^{\frac{\alpha}{2} } }, 0 \!
  \Big)     \! \Big) \! \exp \! \bigg( \! \! -   \! t - 
   \mho_{0,t} \Big(  Q_{\tau,\eta} \Big(\! 0,  \frac{u^{\!\frac{\alpha\epsilon }{2} } }{ v^{\frac{\alpha}{2} } } \!
  \Big) \! \Big) \!\! \bigg) \mathrm{d} t  \mathrm{d} \eta, \hspace{55mm} \mathrm{A} \!=\! \mathrm{Q}\hspace{-2mm}
  \\
   & \! \hspace{-3mm} \widetilde{\mathbf{C}}^{(\mathrm{T})}_{\mathrm{FPC}} \! + \!\! \int^{\tau}_{0} \!\!\!\! \int^{\infty}_{0 } \!\!\!   \mho_{0,t}   \Big( 
   G_{\eta,\eta} \Big(   \frac{u^{\!\frac{\alpha\epsilon }{2} } }{ v^{\frac{\alpha}{2} } }, 0 \!
   \Big) \! \Big) \! \exp \! \bigg( \! \! -  \!  t \!  - \!    \mho_{0,t} \Big(  Q_{\eta,\eta} \Big(\! 0, \!  \frac{u^{\!\frac{\alpha\epsilon }{2} } }{ v^{\frac{\alpha}{2} } } \!
   \Big) \! \Big) \! \!  \bigg) \mathrm{d} t \! 
   \! \int^{\infty}_{0 } \!\!\!\!  \exp \! \bigg( \!\! -  \! t \! - \!    \mho_{0,t}  \Big(\!  Q_{ \tau- \eta,0} \Big(0,  \! \frac{u^{\!\frac{\alpha\epsilon }{2} } }{ v^{\frac{\alpha}{2} } } \!
       \Big) \! \Big)  \!\bigg) \mathrm{d} t  \mathrm{d} \eta, \, \mathrm{A} \!= \!\mathrm{F} \hspace{-2mm}
  \end{dcases}   \vspace{-2mm}
 \end{align}
 \vspace{1pt} \hrulefill
 \end{figure*}
  \begin{figure*} 
  \normalsize  
  \vspace{-5mm}
  \begin{align} 
  & \hspace{-10mm}\widetilde{\mathbf{C}}^{\mathrm{A}}_{\mathrm{FCIPC}}  \! =  \begin{dcases} &   \hspace{-3mm}  \! 
  \widetilde{\mathbf{C}}^{(\mathrm{T})}_{\mathrm{FCIPC}} \!   + \! \!  \int^{\tau}_{0} \!    \mho_{0,1}  \Big(  G_{\tau,\eta}  \Big(   \frac{u^{\!\frac{\alpha}{2}} }{ v^{\!\frac{\alpha}{2}} }\!, 0   \Big) \! \Big) \exp \!  \bigg(  \!\! - \!   \mho_{0,1}  \!   \Big(  Q_{\tau,\eta} \Big(0, \! \frac{u^{\!\frac{\alpha}{2}}}{v^{\!\frac{\alpha}{2}}} \! \Big) \! \Big)   \!    \bigg)  
   \mathrm{d} \eta,  \hspace{39 mm} \mathrm{A}\!=\!\mathrm{Q} \\    \!
  &  \hspace{-3mm}  \widetilde{\mathbf{C}}^{(\mathrm{T})}_{\mathrm{FCIPC}}  \! + \!\! \int^{\tau}_{0} \!   
  \! \mho_{0,1}   \Big(  G_{\eta,\eta}   \Big(   \frac{u^{\!\frac{\alpha}{2}} }{ v^{\!\frac{\alpha}{2}} }\!,   0   \Big) \! \Big) \exp \! \bigg( \! \! - \!   \mho_{0,1} \!   \Big(  Q_{\eta,\eta} \Big(0, \! \frac{u^{\!\frac{\alpha}{2}}}{v^{\!\frac{\alpha}{2}}} \! \Big) \! \Big) \!   
  -   \!  \mho_{0,1}  \! \Big(\! Q_{\tau-\eta,0} \Big(0, \! \frac{u^{\!\frac{\alpha}{2}}}{v^{\!\frac{\alpha}{2}}} \! \Big) \! \Big) \!   \bigg)   
  \mathrm{d} \eta
  , \hspace{5mm} \mathrm{A}\!=\!\mathrm{F}
  \end{dcases} \hspace{-2mm} \label{FCIPC}
  \end{align}
  \vspace{0pt} \hrulefill \vspace{-2mm}
  \end{figure*} 
\end{corollary}

With the analytical expressions of $\widetilde{\mathbf{C}}^{\mathrm{A}}_{\mathrm{GFPC}} $ and $ \widetilde{\mathbf{C}}^{\mathrm{A}}_{\mathrm{FPC}} $ in (\ref{eqn:M_II}) and (\ref{C1:FPC}), respectively, we have the following observation.

\noindent
{\bf Remark 1}: In an ultra-dense network, i.e. when $\zeta_{B} \to \infty$, $\widetilde{\mathbf{C}}^{\mathrm{A}}_{\mathrm{GFPC}}  \sim \widetilde{\mathbf{C}}^{\mathrm{A}}_{\mathrm{FPC}}   $, 
which can be verified by letting $B_{\widehat{P}} \to \infty $ as a direct result of $\zeta_{B} \to \infty$.

\begin{corollary} \label{C3}
The uplink coverage probability with NPC (i.e. when $\epsilon=0$) is closely approximated by (\ref{NPC}), where $\widetilde{\mathbf{C}}^{(\mathrm{T})}_{\mathrm{NPC}} \! = \! \int^{\infty }_{0} \!  \exp \! \big(   \! - \! t  \! -     \varphi_{t} \big( Q_{\tau,0} (0,v^{-\frac{\alpha}{2}}) \big)  \big) \mathrm{d} t$ and  $\varphi_{t}(z) \! := \! \frac{t}{C_{2}}　\int^{\infty}_{0} \! \!  z  \big ( 1 - e^{- \frac{C_{1}}{C_{2}}  v t  } \big) \mathrm{d} v $.

\begin{figure*}    \vspace{-3mm}
\begin{align} \label{NPC}
  \hspace{-1mm} & \widetilde{\mathbf{C}}_{\mathrm{NPC}}\!=\!\!  
  \begin{dcases} & \hspace{-4mm} \widetilde{\mathbf{C}}^{(\mathrm{T})}_{\mathrm{NPC}} \! + \!\!  \int^{\tau}_{\!0} \!\!\!\! \int^{\infty}_{\!0} \!\!\!\!   \varphi_{t} \big( G_{\!\tau,\eta}(v^{-\frac{\alpha}{2}},0 )  \big)   \!     
 \exp \! \Big(  \! \! - \! t    \! -   \varphi_{t} \big( Q_{\!\tau,\eta}(0,  v^{-\frac{\alpha}{2}} \! )  \big) \! \Big) \mathrm{d}  t  \mathrm{d} \eta , \hspace{53mm} \mathrm{A} \!=\! \mathrm{Q}   \\
 & \hspace{-4mm} \widetilde{\mathbf{C}}^{(\mathrm{T})}_{\mathrm{NPC}} \! + \! \!\! \int^{\tau}_{0} \!\!\!\! \int^{\infty}_{0} \!\!\!\!  \varphi_{t} \big(G_{\!\eta,\eta}(v^{-\frac{\alpha}{2}}\!,0)  \big)  \!     
  \exp \! \Big(  \! \! - \! t   \! - \!  \varphi_{t} \big(Q_{ \eta,\eta }(0,v^{-\frac{\alpha}{2}}\!)  \big)   \! \Big)  \mathrm{d}  t \! \! \int^{\infty}_{0} \!\!\!\! \exp \! \Big(\! \! - \! t - \!   \varphi_{t} \big(Q_{ \tau-\eta,0}(0,v^{-\frac{\alpha}{2}})  \big)  
\!   \Big) \mathrm{d}  t  \mathrm{d} \eta , \, \mathrm{A}\!=\! \mathrm{F}
\end{dcases} \hspace{-4mm}
\end{align} 
\hrulefill \vspace{-5mm}
\end{figure*}

\end{corollary} 
With the analytical expression of  $\widetilde{\mathbf{C}}^{\mathrm{A}}_{\mathrm{NPC}}$ in  (\ref{NPC}), we have another scaling property of   $\widetilde{\mathbf{C}}^{\mathrm{A}}_{\mathrm{GFPC}}$ in (\ref{eqn:M_II}) as follows.

\noindent
{\bf Remark 2}: Given the SIR threshold $\tau$ and $\epsilon$, when $\zeta_{B} \to 0$,  $\widetilde{\mathbf{C}}_{\mathrm{GFPC}}  \sim \widetilde{\mathbf{C}}_{\mathrm{NPC}}  $, which can be verified by letting $B_{\widehat{P}} \to 0 $ as a direct result of $\zeta_{B} \to 0$.

\begin{corollary} \label{C2}  The uplink coverage probability with TFPC (i.e. when $\bar{P}=0$) is approximated by  \begin{align}
&\hspace{-3mm} \widetilde{\mathbf{C}}^{\mathrm{A}}_{\mathrm{TFPC}}   \! = \!\! \begin{dcases}  & \hspace{-3mm} \widetilde{\mathbf{C}}^{(\mathrm{T})}_{\mathrm{TFPC}}+\int^{\tau}_{0} \!\!\! \int^{B_{\!\widehat{P}} }_{0} \!  \! \xi_{\tau,\eta,t} \Big(   \frac{u^{ \! \frac{\alpha \epsilon }{2}    } }{  v^{ \! \frac{\alpha}{2}  } },0  \Big)   \mathrm{d}  t \mathrm{d}  \eta , \quad \mathrm{A} = \mathrm{Q}  \\
& \hspace{-3mm} \widetilde{\mathbf{C}}^{(\mathrm{T})}_{\mathrm{TFPC}}+ \int^{\tau}_{0} \!\!\! \int^{B_{\!\widehat{P}} }_{0} \!\!\! \exp \! \Big(\!- t -  \omega_{\tau-\eta,0,t} \Big(  \frac{u^{ \! \frac{\alpha \epsilon }{2}    } }{  v^{ \! \frac{\alpha}{2}  } }, 0  \Big)
    \!  \Big)     \mathrm{d}  t \nonumber \\ 
& \hspace{10mm} \times \int^{B_{\!\widehat{P}} }_{0} \!  \!   \xi_{\eta,\eta,t} \Big(  \frac{u^{ \! \frac{\alpha \epsilon }{2}    } }{  v^{ \! \frac{\alpha}{2}  } }, 0  \Big)    \mathrm{d}  t    \mathrm{d}  \eta,  \! \quad \quad \mathrm{A} = \mathrm{F}  \hspace{-10mm}
\end{dcases}
\end{align}
where $\omega$ and $\xi$ are defined in (\ref{omega}) and (\ref{xi}), respectively, and $\widetilde{\mathbf{C}}^{(\mathrm{T})}_{\mathrm{TFPC}} = \int^{\!B_{\!\widehat{P}} }_{0} \!\!   \exp \! \Big (\! - \! t \! -   \!   
\omega_{\tau,0,t} \big( u^{ \frac{\alpha\epsilon}{2} }   v^{-\frac{\alpha}{2}} , 0 \big )
\!   \Big ) \mathrm{d}t$.

  
\end{corollary}

\begin{figure*}  
 \centering
  \vspace{-2mm}  
  \subfigure [$\epsilon$ = 0 (NPC, {\bf Corollary}  \ref{C3}) \vspace{-0mm}]
   {
 \label{fig:I_1}
  \centering   
  \includegraphics[width=0.32 \textwidth]{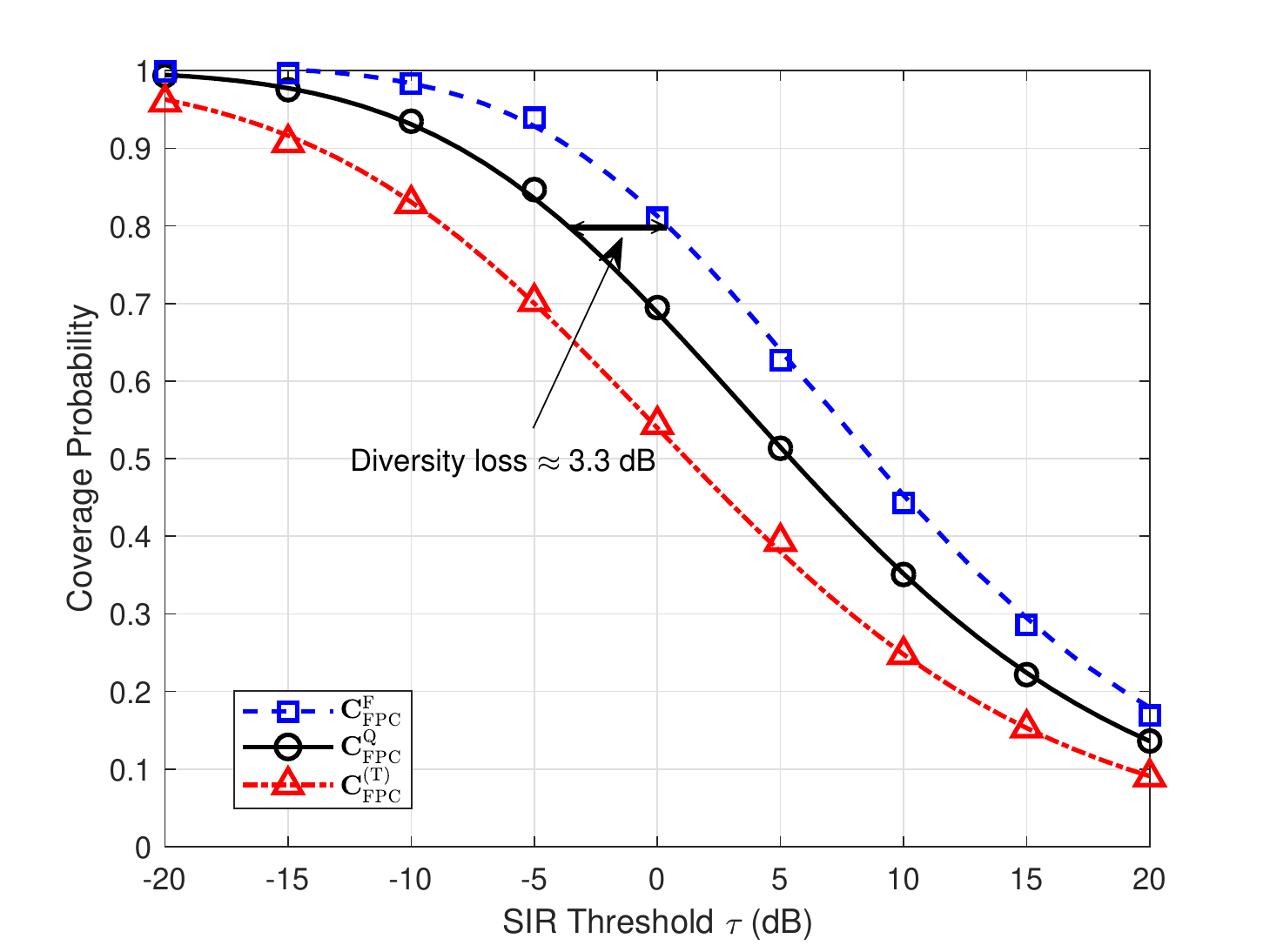}}
  \centering  
  \subfigure  [ 
$\epsilon$ = 0.5 (FPC,  {\bf Corollary}  \ref{C1}) ] {
 \label{fig:I_2}
  \centering
 \includegraphics[width=0.32 \textwidth]{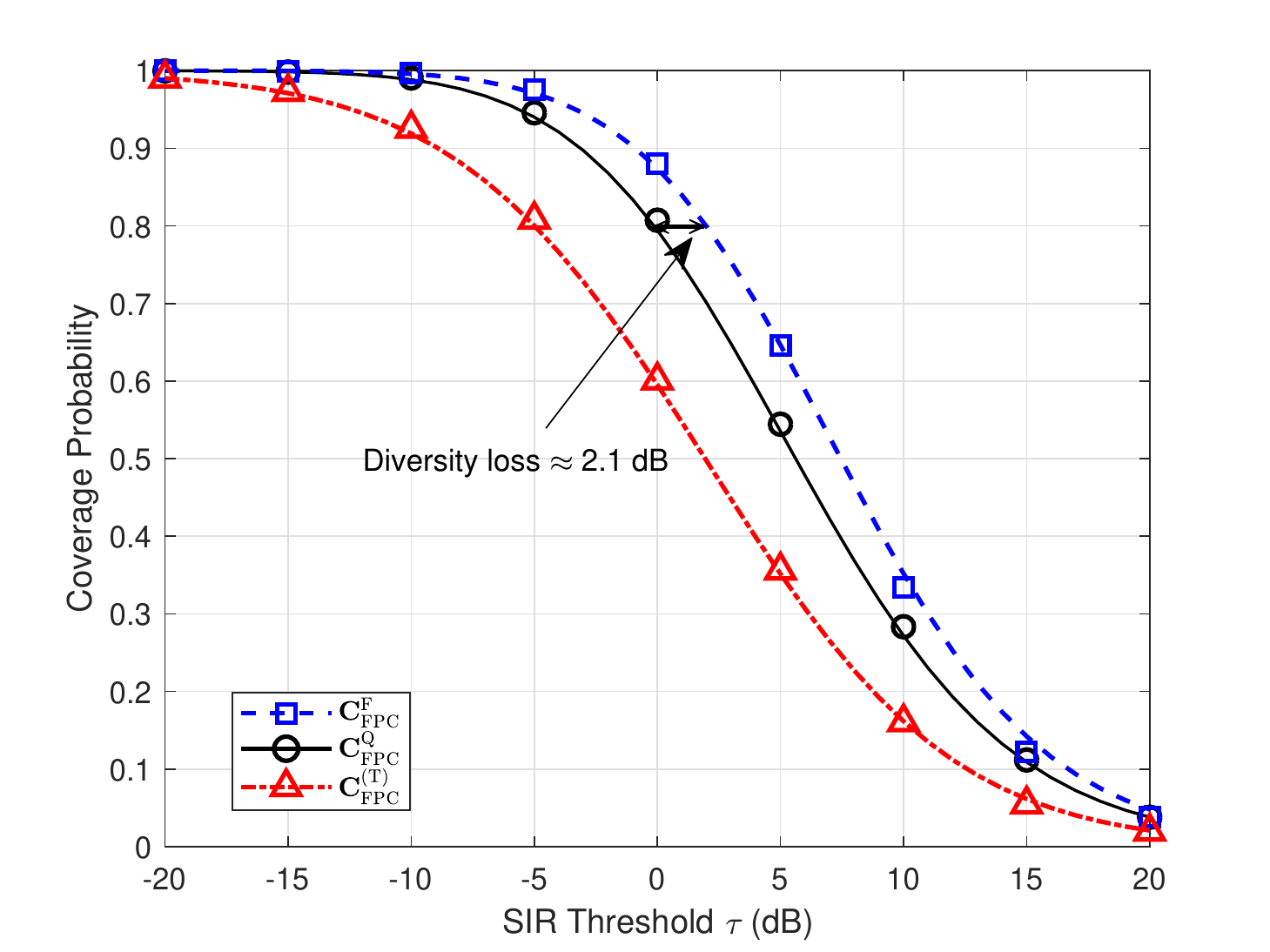}}
  \vspace{-0mm}
  \centering
    \subfigure [$\epsilon$ = 1 (FCIPC,  {\bf Corollary}  \ref{C1}) \vspace{-0mm}]
     {
   \label{fig:I_3}
    \centering   
    \includegraphics[width=0.32 \textwidth]{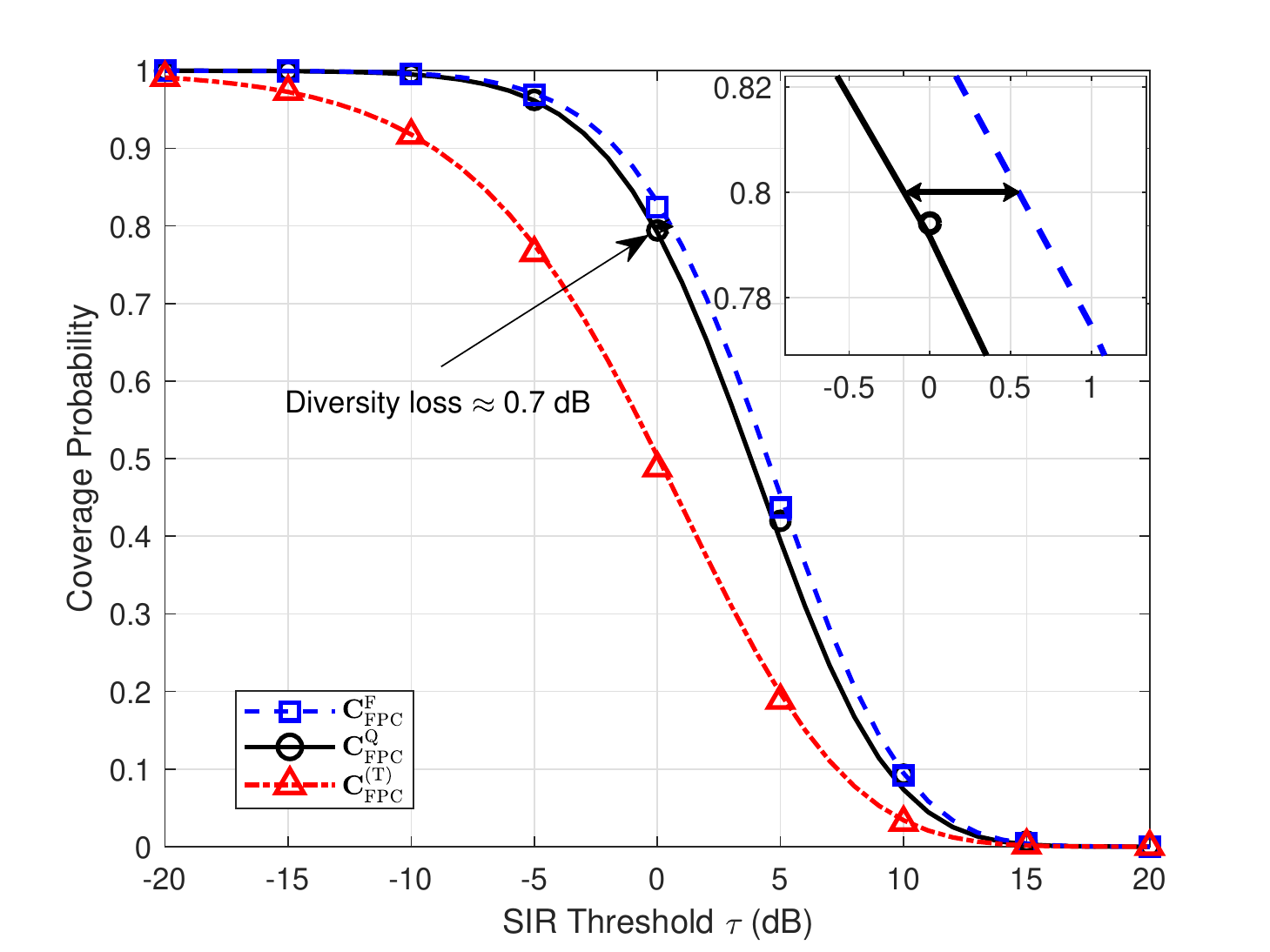}} \vspace{-2mm}
 \caption{The coverage probability with different values of $\epsilon$. 
 } 
 \centering
 \label{fig:tauAB}
 \vspace{-0mm}
 \end{figure*}

 \begin{figure*}  \vspace{-4mm}
 \centering
   \begin{minipage}[c]{0.325 \textwidth}
    \includegraphics[width=0.95\textwidth]{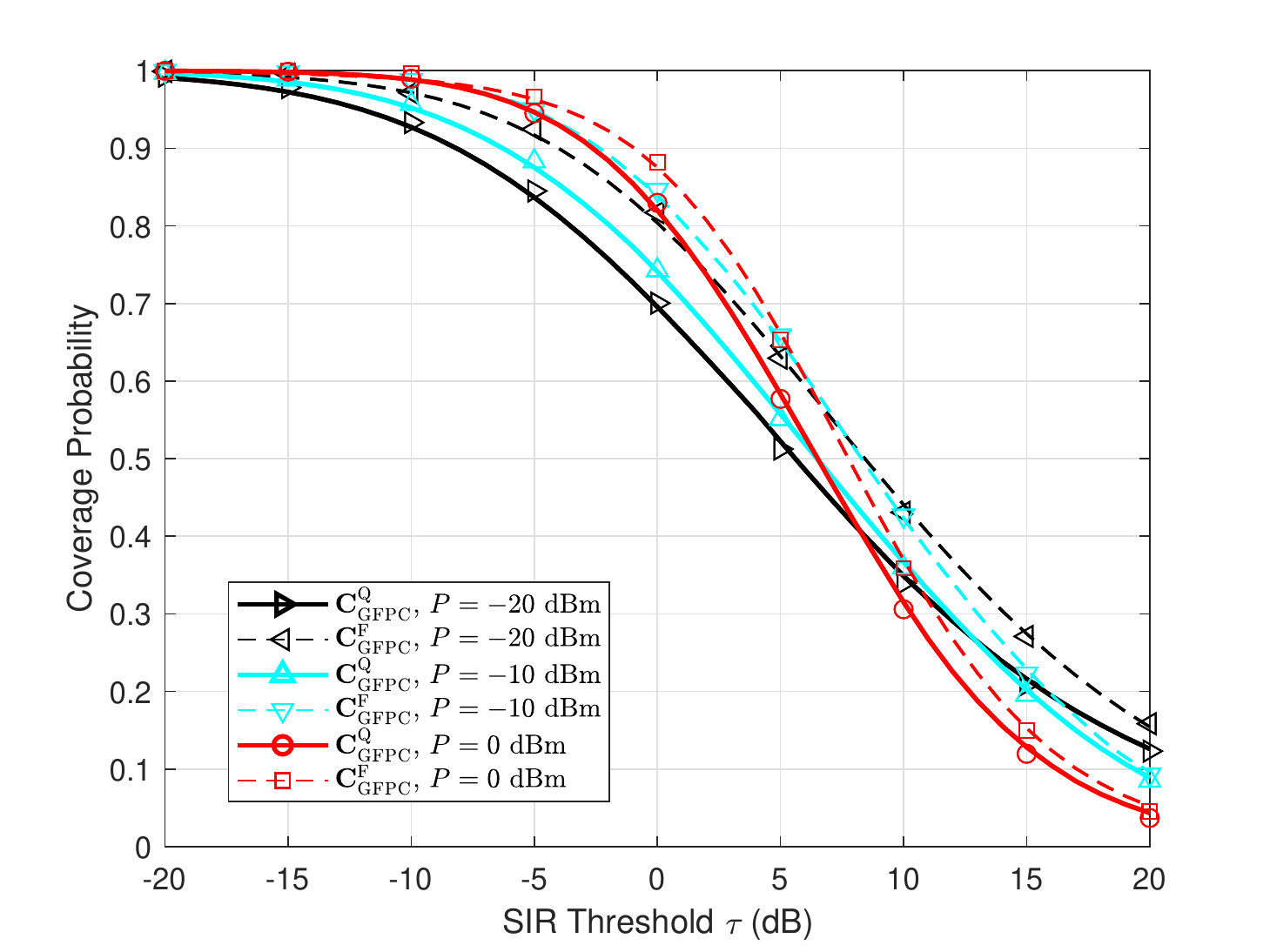} \vspace{-3mm}
    \caption{The coverage probability with different values of $\widehat{P}$  (for $\bar{P}=\widehat{P}$).
    } \label{fig:PM}
   \end{minipage}
   \begin{minipage}[c]{0.33\textwidth}
   \includegraphics[width=0.95\textwidth]{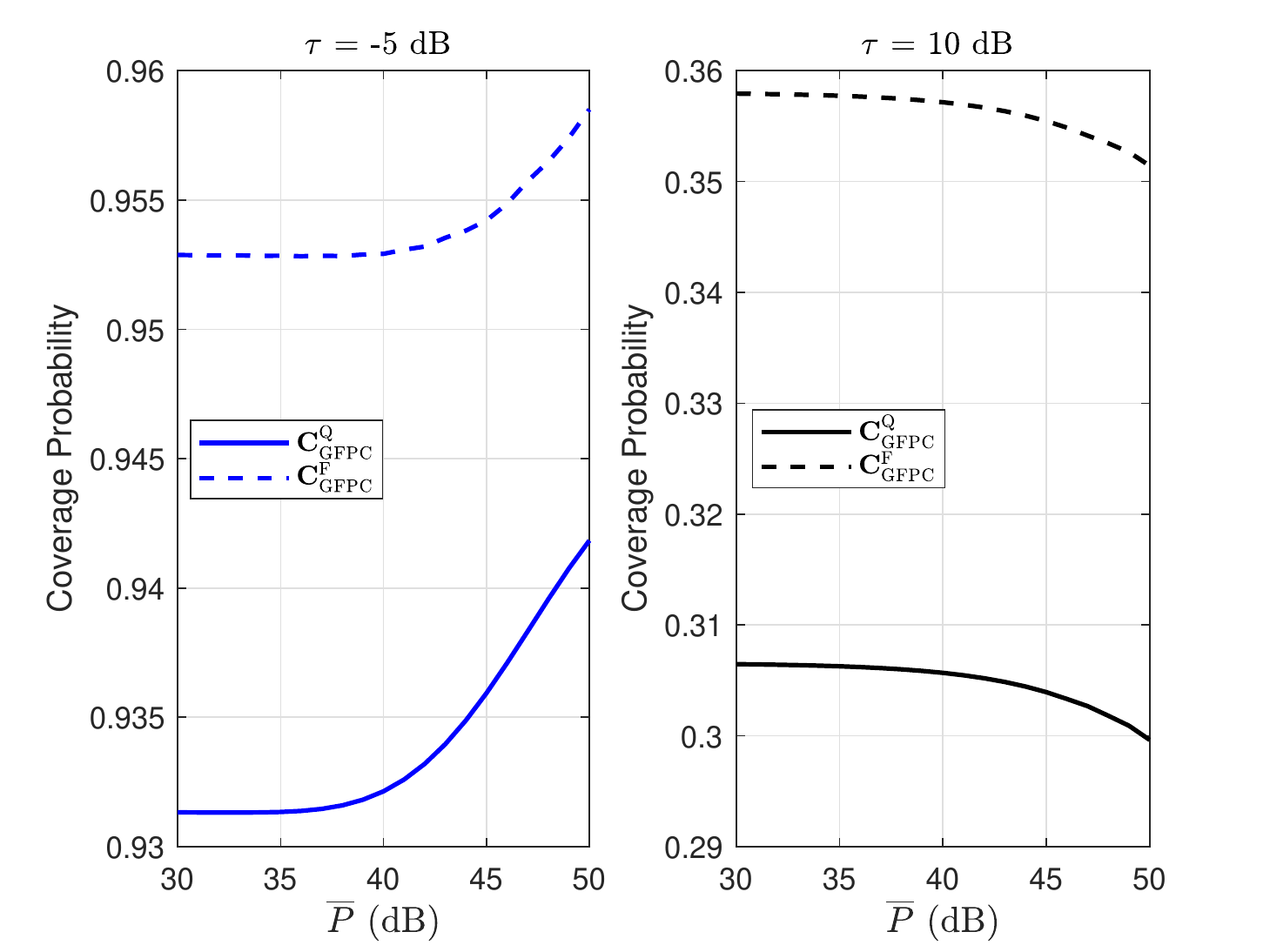} \vspace{-3mm}
   \caption{The coverage probability with different values of $\bar{P}$ (for $\epsilon$ = 0.5). } \label{fig:Pt}
   \end{minipage}
   \begin{minipage}[c]{0.325 \textwidth}\vspace{0mm}
     \includegraphics[width=0.95\textwidth]{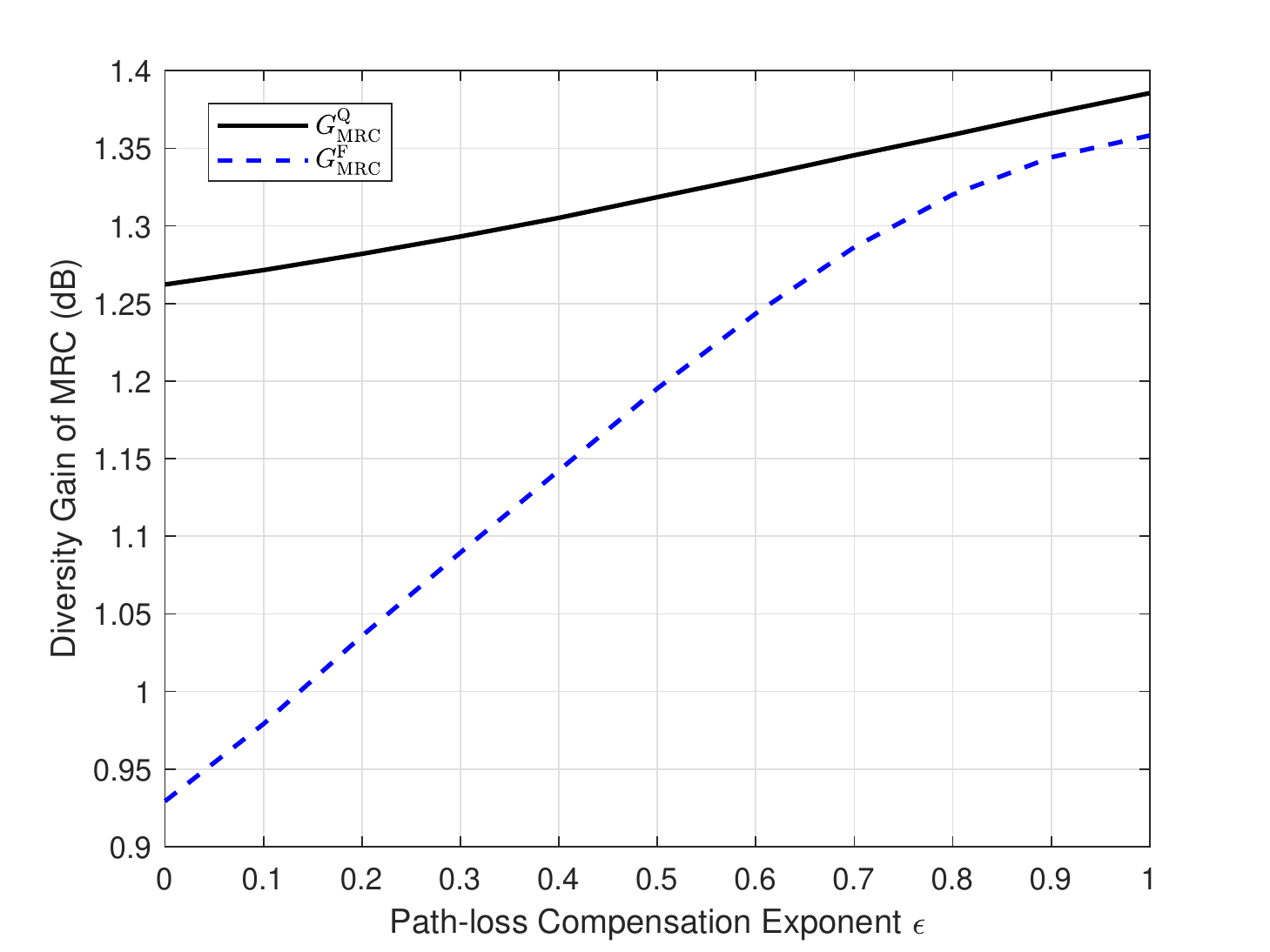} \vspace{-3mm}
     \caption{Diversity gain of MRC with different values of $\epsilon$ (for $\tau=0$ dB, $\widehat{P}=\infty$).
     } \label{fig:diversitygain}
     \end{minipage}  \vspace{-4mm}
  \end{figure*} 
  
\section{Numerical Results}
This section validates our analytical expressions through Monte Carlo simulations and evaluates the impact of different power control parameters. 
In the simulations, we set the BS density $\!\zeta_{B}\!$ as $\!10$ BSs/km$^2$, path-loss exponent $\alpha$ as $4$, and baseline transmit power $\varrho$ as $\!-50\!$ dBm, unless otherwise stated. The curves and the markers are used to represent the analytical results and the simulation results, respectively.

\subsubsection{Impact of $\epsilon$}
Figure~\ref{fig:tauAB} compares the coverage probabilities achieved by the retransmission scheme with FVI and QSI and the initial transmission  
under FPC.  The performance degradation of the retransmission scheme with QSI over that with FVI  can be observed under different settings of PCE. To quantify such performance degradation, 
we evaluate the {\em retransmission diversity loss}, defined as the ratio of the required SIR to achieve a target coverage probability with FVI to that with QSI, i.e., $ \frac{f^{-1}(\mathbf{C}^{\mathrm{Q}}_{\mathrm{FPC}}(\tau)) }{f^{-1} (\mathbf{C}^{\mathrm{F}}_{\mathrm{FPC}}(\tau) )}$, where $f^{-1}(\mathbf{C}^{\mathrm{A}}_{\mathrm{FPC}}(\tau)) $ represents the inverse function of  $\mathbf{C}^{\mathrm{A}}_{\mathrm{FPC}}(\tau) $.
When the target coverage probability is $80\%$, the diversity loss is around $3.3$ dB, $2.1$ dB and $0.7$ dB, when PCE $\epsilon$ equals $0$, $0.5$ and $1$, respectively. This reveals that the diversity loss can be effectively mitigated by increasing the path-loss compensation. Another observation is that 
higher path-loss compensation increases the coverage probability with a low SIR threshold, however, decreases the coverage probability with a large SIR threshold.
Thus, it is more beneficial to adopt a smaller power control exponent when the target SIR is large.  


\subsubsection{Impact of $\widehat{P}$}  Fig.~\ref{fig:PM} illustrates the coverage probabilities when the maximum transmit power $\widehat{P}\!= \!-20 $ dBm, $-10$ dBm, $0$ dBm, (i.e. $\widehat{P}/\varrho = 30$ dB, $40$ dB, $50$ dB). It can be observed that the coverage probability benefits from larger $\widehat{P}$ at high-coverage regime but smaller  $\widehat{P}$ at the low-coverage regime.
When the target coverage probability is $80\%$, the diversity loss $L_{\Xi}$ is $4.3$ dB, $3.7$ dB and $1.7$ dB, respectively. This implies that setting a larger maximum transmit power reduces the diversity loss, which agrees with the findings in Fig.~\ref{fig:tauAB}. 

\subsubsection{Impact of $\bar{P}$}
In Fig.~\ref{fig:Pt}, we examine how the enforced transmit power $\bar{P}$ affects the coverage probabilities when the SIR threshold is relatively small (i.e. $\tau=-5$ dB) and large (i.e. $\tau=10$ dB).   
We find that larger $\bar{P}$ increases the coverage probabilities when $\tau$ is small but decreases the coverage probabilities when $\tau$ is large. 
This can be understood from the fact that the SIR at the receiver is more dominated by the received signal power and the interference power when $\tau$ is small and large, respectively. A larger $\bar{P}$ brings higher received signal power while smaller $\bar{P}$ causes less interference.

\subsubsection{Effect of MRC} Finally, we study the effect of MRC under FPC by comparing the coverage probability under retransmissions with and without MRC. The benefit of MRC is quantified in terms of {\em diversity gain} defined as  $G^{\mathrm{A}}_{\mathrm{MRC}} \!\! : = \!\! \frac{f^{-1}(\mathbf{C}^{\mathrm{A}} ) }{f^{-1}(\mathbf{T}^{\mathrm{A}} )}$, where $\mathbf{T}^{\mathrm{A}}$ denotes the coverage probability under retransmission without MRC (i.e. Type-I HARQ). 
Fig.~\ref{fig:diversitygain}
demonstrates the diversity gain of MRC under FPC with different values of PCE\footnote{The coverage probability of retransmission without MRC under FPC can be generated from (\ref{C1:FPC}) with $G_{\tau,\eta}$ and $Q_{\tau, \eta}$ replaced by $G_{\tau+\eta,\eta}$ and $Q_{\tau+\eta,\eta}$, respectively, when $\mathrm{A}=\mathrm{Q}$ and $Q_{\tau-\eta,0}$ replaced by $Q_{\tau,0}$ when $\mathrm{A}=\mathrm{F}$. }.                    
The results manifest that larger diversity gain of MRC can be achieved by increasing the transmit power. Moreover, 
the diversity gain of MRC is more pronounced in the scenario with QSI than that with FVI, especially when the transmit power is low.

 \vspace{-2mm}    
\section{Conclusion}
This letter presents a stochastic geometry analysis of a Type-II HARQ-CC retransmission scheme in uplink cellular networks with transmit power control. 
In particular, the uplink coverage probability in a large-scale cellular network  is characterized in the scenarios of both QSI and FVI. Our study reveals the effects of different power control parameters on the retransmission diversity loss and the diversity gain achieved by MRC. 
In addition, the derived analytical expressions can be utilized to maximize the uplink coverage probability by optimizing the power control parameters.

\section*{Appendix}
\section*{Proof of Theorem 1} \vspace{-1mm}
\begin{proof}   

The distribution of uplink users is not a PPP. For tractability, we adopt the method introduced in \cite{M.April2017Haenggi} to 
approximate the distribution of the users. Specifically, $\Phi_{I}$ are modeled by an inhomogeneous PPP with density $\zeta_{I}$. The  functions of the density $\zeta_{I}$, the PDF of $l_{0}$, i.e., $f_{l_{0}}(r)$, and that of $l_{j}, \forall j \in \Phi_{I}$, i.e., $f_{l_{j}}(l)$, are given, respectively, as 
\begin{align}  
\zeta_{I}(d_{j}) = \zeta_{B} \Big(1-\exp\big(C_{1} \zeta_{B} \pi\|d_{j}\|^2
\big)\Big), \nonumber
\end{align}
\vspace{-2mm}
\begin{align} 
  f_{l_{0}}(r)= 2 C_{2} \pi \zeta_{B} r \exp \Big(\! -   C_{2} \zeta_{B} \pi r^2 \Big), \nonumber
\end{align}  
and 
\begin{align}  
 \hspace{-3mm} \quad f_{l_{j}|d_{j}}(r) = \frac{2C_{2} \pi \zeta_{B} r \exp \big( \! - C_{2} \pi \zeta_{B} r^2 \big)  }{1-\exp\big(\! - C_{2}  \pi \zeta_{B} d^2_{j} \big)}, \, 0 \leq r \leq d_{j}, \nonumber
\end{align} 
where $C_{1}=\frac{12}{5}$ and $C_{2}=\frac{13}{10}$.

From the definition in (\ref{def}), the coverage probability  
can be expressed as
 \begin{align}
  \mathbf{C}^{\mathrm{A}}_{\mathrm{GFPC}}  \! =  \!  \mathbf{C}^{(\mathrm{T})}_{\mathrm{GFPC}} 
 \! + \! \mathbb{E} 
 \big[   
 \mathbb{P} [ \eta^{(\mathrm{R})} \! > \!  \tau \!- \! \eta^{(\mathrm{T})} | \Xi^{\mathrm{A}} ]  \mathbbm{1}_{ \{\eta^{(\mathrm{T})} \!  \leq    \tau  \} }      \big] , \nonumber  \hspace{-5mm}   
 \end{align}
where $\mathbb{E}$ denotes the expectation operator, $\mathbbm{1}_{ \{ \cdot \}}$ is the inductor function, $\mathbf{C}^{(\mathrm{T})}_{\mathrm{GFPC}} $ represents the coverage probability of an initial transmission, 
which is obtained as (\ref{eqn:S}) following similar derivations to 
the proof of \textbf{Theorem 1} in~\cite{Y.2017Wang}, and the second term in the above equation can be derived as
\begin{align}
& \mathbb{E} 
 \big[   
 \mathbb{P} [ \eta^{(\mathrm{R})} \!> \!  \tau \!- \! \eta^{(\mathrm{T})}  ]  \mathbbm{1}_{ \{\eta^{(\mathrm{T})}   \leq    \tau  \} }      \big] \nonumber \\
& \! \! = \! \mathbb{E}_{\Xi^{\mathrm{A}}\!,\bm{h}^{(\!\mathrm{R}\!)}   } \! \Bigg[ \mathbb{P} \bigg[ \frac{h_{0}^{\!(\mathrm{R})} 
  l^{\alpha (\epsilon-1 )  }_{0} }{ \sum_{j \in \Phi^{(\mathrm{R})}_{\!I} }  \! h_{j}^{\!(\mathrm{R})}  
  l^{\alpha \epsilon  }_{j} d^{-\alpha}_{j}   } \!  > \! \tau \! - \! \eta^{(\mathrm{T})}  \!  \bigg] \! \mathbbm{1}_{    \{  \eta^{(\mathrm{T})}   \leq    \tau \!  \} }   \! \Bigg]  \nonumber \\
& \!\! \overset{\text{(b)}}{=} \!  
\mathbb{E}_{\Xi^{\mathrm{A}}\!}\Bigg[\!  \int^{\tau}_{0} \!\!\!\!  \prod_{j \in \Phi^{(\mathrm{R})}_{\!I}} \! \!\!\!  \Big(   1  \! + \! (\tau  \! - \!  \eta^{(\mathrm{T}) } \! )  \frac{  l_{0}^{ \alpha  } l_{j}^{\alpha \epsilon  } }{ d^{ \alpha }_{j} l_{0}^{ \alpha  \epsilon   }  } \Big)^{\!\!-1} \! \!   f_{\eta^{(\mathrm{T})}  }(\eta) \mathrm{d} \eta \Bigg]
, \!\! \label{eqn:CPP_II}
\end{align}
where (b) holds as $h_{i}, \forall i \in \Phi_{U}$ is exponentially distributed and  $f_{\eta^{(\mathrm{T})}  }(\eta)$ represents the PDF of $\eta^{(\mathrm{T})}$     
 calculated as  (\ref{eq:II_PDF}) shown on the top of the next page.
\begin{figure*}
\begin{align}  
 &\hspace{-2mm}  f_{ \eta^{(\mathrm{T})}    } (\eta)  
 =  \partial \mathbb{P} [  \eta^{(\mathrm{T})} \! \leq \! \eta \big    ]  /\partial \eta   
  =  \partial \bigg( \! 1  \! - \! \mathbb{E}_{\bm{h}^{(\mathrm{T} ) } } \! \bigg[ \exp \! \Big( \! - \! \eta \!  \sum_{j \in \Phi^{(\mathrm{T})}_{\!I}} \!\! h_{j}^{\! (\mathrm{T})}  l_{0}^{ \alpha (1-\epsilon ) } l_{j}^{\alpha \epsilon  } d^{-\alpha }_{j}   \Big) \bigg] \bigg) / \partial \eta  \nonumber \\
& \hspace{-2mm} = \partial \Big( 1 - \prod_{j \in \Phi^{(\mathrm{T})}_{\! I}} \!  \!  \Big( 1   +  \eta  l_{0}^{ \alpha (1-\epsilon ) } l_{j}^{\alpha \epsilon  } d^{-\alpha }_{j}    \Big)^{\!\!-1} \Big) / \partial \eta    = \!\!  \sum_{j \in \Phi^{(\mathrm{T})}_{\!I} } \! 
\frac{l_{0}^{ \alpha( 1 - \epsilon ) } l_{j}^{\alpha \epsilon  } d^{-\alpha}_{j}}{( 1 \! + \! \eta l_{0}^{ \alpha( 1 - \epsilon ) } l_{j}^{\alpha \epsilon  } d^{-\alpha}_{j} )^{2}}  \!\! \prod_{  j^{\prime} \in \Phi^{\!(\mathrm{T})}_{\!I}\!  \backslash \{j\}  }  \!\!\!\!\!\!\!\! \Big( 1 \! + \! \eta \frac{l^{ \alpha  }_{0}  l^{  \alpha \epsilon} _{j^{\prime}} }{d^{\alpha}_{j^{\prime}}  l^{\alpha\epsilon}_{0} }   \! \Big)^{\!\!-1} \!\!  .  \hspace{-2mm}    
 \label{eq:II_PDF}
  \vspace{-0mm}
\end{align}    \vspace{-0mm}
\hrulefill
\end{figure*}

\begin{figure*}
\begin{align}  
 &  \hspace{-2mm} 
 \mathbb{E} 
  \Bigg[   \!   \int^{ \infty}_{ 0} \!\!\!\!  \prod_{k \in \Phi^{(\mathrm{R})}_{\!I}} \!\!\!\! \Big(  1 \! + \! (\tau \! - \! \eta   )   
\frac{l^{ \alpha  }_{0}  l^{  \alpha \epsilon} _{k} }{d^{ \alpha}_{k}  l^{\alpha\epsilon}_{0} }   \Big)^{\!\!-1}   \!\!\! \sum_{j \in \Phi^{ (\mathrm{T})}_{\!I} } \! \!  \frac{l_{0}^{ \alpha( 1 - \epsilon ) } l_{j}^{\alpha \epsilon  } d^{-\alpha}_{j}}{( 1 \! + \! \eta l_{0}^{ \alpha( 1 - \epsilon ) } l_{j}^{\alpha \epsilon  } d^{-\alpha}_{j} )^{2}} 
\!\!\!   \prod_{   j^{\prime} \in \Phi^{(\mathrm{T})}_{\!I}\! \backslash \{j\}  }  \!\!\!\!\!\!   \Big( 1 \! + \! \eta 
\frac{l^{ \alpha  }_{0}  l^{  \alpha \epsilon}_{j^{\prime}} }{d^{ \alpha}_{j^{\prime}}  l^{\alpha\epsilon}_{0} }  \Big)^{\!\!-1} \!  \mathrm{d} \eta \Bigg]  \nonumber\\
& \hspace{-2mm} \overset{\text{(c)}}{=} \! \begin{dcases}  \int^{\tau}_{0}    \!   \mathbb{E} 
\Bigg[ \sum_{j \in \Phi _{\!I} } \!  
g_{\eta} \Big(   \frac{l^{ \alpha  }_{0}  l^{ \alpha \epsilon}_{j} }{d^{ \alpha}_{j}  l^{\alpha\epsilon}_{0} }  \! \Big)  
 f_{ \tau-\eta }   \Big(    
\frac{l^{ \alpha  }_{0}  l^{  \alpha \epsilon} _{j} }{d^{ \alpha}_{j}  l^{\alpha\epsilon}_{0} } \!  \Big)  \mathbb{E}_{\Xi^{\mathrm{A}} } \bigg[ \prod_{  k \in \Phi _{\!I}\! \backslash \{j\}   }  \!\!\! f_{\eta} \Big( \frac{l^{ \alpha  }_{0}  l^{  \alpha \epsilon} _{k} }{d^{ \alpha}_{k}  l^{\alpha\epsilon}_{0} }    \Big)  f_{ \tau-\eta }   \Big(    
\frac{l^{ \alpha  }_{0}  l^{  \alpha \epsilon} _{k} }{d^{ \alpha}_{k}  l^{\alpha\epsilon}_{0} }  \! \Big) \! \bigg] \!  \Bigg]     
\mathrm{d} \eta , \quad \mathrm{A} = \mathrm{Q}   \\
\int^{\tau}_{0}    \!   \mathbb{E} 
\Bigg[ \sum_{ j \in \Phi^{(\mathrm{T})}_{\!I} } \! \!  
g_{\eta} \Big(   \frac{l^{ \alpha  }_{0}  l^{  \alpha \epsilon} _{j} }{d^{ \alpha}_{j}  l^{\alpha\epsilon}_{0} }  \! \Big)  
  \mathbb{E}_{\Xi^{\mathrm{A}} } \bigg[ \!\!  \prod_{  j^{\prime} \in \Phi^{(\mathrm{T})}_{\!I}\! \backslash \{j\}   } \!\! \!\!\! f_{\eta} \Big( \frac{l^{ \alpha  }_{0}  l^{  \alpha \epsilon} _{j^{\prime}} }{d^{ \alpha}_{j^{\prime}}  l^{\alpha\epsilon}_{0} }    \Big) \!\!\!  \prod_{  k \in \Phi^{(\mathrm{T})}_{\!I}    } \!\!  f_{ \tau-\eta }   \Big(    
\frac{l^{ \alpha  }_{0}  l^{  \alpha \epsilon} _{k} }{d^{ \alpha}_{k}  l^{\alpha\epsilon}_{0} }  \! \Big)   \bigg]  \! \Bigg]     
\mathrm{d} \eta , \hspace{12mm} \mathrm{A} = \mathrm{F}   \\
\end{dcases} \nonumber \\
&  \hspace{-2mm} \overset{\text{(d)}}{=}   \begin{dcases}  \int^{\tau}_{0}\!     \mathbb{E}_{l_{0}} \! \Bigg[  
2 \pi\!\! \int^{\infty}_{0} \!\!\! \zeta_{I}(x)    \! \int^{\infty}_{x} \!\!         
g_{\eta} \Big(    \frac{l^{ \alpha  }_{0}  l^{  \alpha \epsilon}  }{x^{ \alpha}   l^{\alpha\epsilon}_{0} } \! \Big)  f_{ \tau-\eta }   \Big(    \frac{l^{ \alpha  }_{0}  l_{k}^{  \alpha \epsilon}  }{x^{ \alpha}   l^{\alpha\epsilon}_{0} }  \!  \Big)   f_{l_{j}}(l)  \mathrm{d} l \mathrm{d} x \cdot  \underbrace{\mathbb{E}_{\Xi^{\mathrm{A}}} \! \bigg[ \!   
  \prod_{  k \in \Phi^{(\mathrm{T})}_{\!I} \! \backslash \{j\}   }  \!\!\!\!\! f_{\eta} \Big(   \frac{l^{ \alpha  }_{0}  l_{k}^{  \alpha \epsilon}  }{d_{k}^{ \alpha}   l^{\alpha\epsilon}_{0} }  \! \Big)  f_{ \tau-\eta }   \Big(    \frac{l^{ \alpha  }_{0}  l_{k}^{  \alpha \epsilon}  }{d_{k}^{ \alpha}   l^{\alpha\epsilon}_{0} }  \!  \Big)   \! \bigg]}_{:= \mathbf{S}^{\mathrm{Q}}} \! \Bigg]   
 \mathrm{d} \eta,  \quad \!\!\!\! \mathrm{A} = \mathrm{Q} \hspace{-2mm}   \\ %
 \int^{\tau}_{0}    \!    \mathbb{E}_{l_{0}} \!  
 \Bigg[ 2 \pi\!\! \int^{\infty}_{0} \!\!\! \zeta_{I}(x)    \! \int^{\infty}_{x} \!\!         
 g_{\eta} \Big(    \frac{l^{ \alpha  }_{0}  l^{  \alpha \epsilon}  }{x^{ \alpha}   l^{\alpha\epsilon}_{0} } \! \Big)    f_{l_{j}}(l)  \mathrm{d} l \mathrm{d} x \cdot    \underbrace{\mathbb{E}_{\Xi^{\mathrm{A}}} \! \bigg[ \!   
   \prod_{  j^{\prime} \in \Phi^{(\mathrm{T})}_{\!I} \!  \backslash \{j\}   }  \!\!\!\!\! f_{\eta} \Big(   \frac{l^{ \alpha  }_{0}  l_{j^{\prime}}^{  \alpha \epsilon}  }{d_{j^{\prime}}^{ \alpha}   l^{\alpha\epsilon}_{0} }  \! \Big) \bigg]  \mathbb{E}_{\Xi^{\mathrm{A}}} \! \bigg[    \prod_{  k \in \Phi^{(\mathrm{T})}_{\!I}     }  \!\! f_{ \tau-\eta }   \Big(    \frac{l^{ \alpha  }_{0}  l_{k}^{  \alpha \epsilon}  }{d_{k}^{ \alpha}   l^{\alpha\epsilon}_{0} }  \!  \Big)   \! \bigg] }_{ := \mathbf{S}^{\mathrm{F}}}   \! \Bigg]    
  \mathrm{d} \eta,  \quad \!\! \mathrm{A} = \mathrm{F} \hspace{-5mm}  
 \end{dcases} \hspace{-5mm} \label{eqn:C_3} 
 \vspace{-0mm}
\end{align}
\vspace{0pt} \hrulefill
\end{figure*}  
  
Plugging (\ref{eq:II_PDF}) into  
(\ref{eqn:CPP_II}) and applying the transformations $u \! = \! (\frac{ l  }{ l_{0}})^2$ and $v\!= \!(\frac{d}{l_{0}})^2$ yields (\ref{eqn:C_3}) shown on the next page, where (c) follows the substitution $f_{a}(x)\!=\!(1+ax)^{-1}$ and $g_{a}(x) \!=\! x(1+ax)^{-2}$, (d) applies the Campbell Mecke  formula~\cite{M.2013Haenggi}, and by following the probability generating functional for a PPP, $\mathbf{S}^{\mathrm{A}}$ can be further derived as \vspace{2mm} 
\begin{align} 
& \hspace{-2mm} \mathbf{S}^{\mathrm{A}} \!\! = 
\!\! \begin{dcases} \!  
\exp \! \bigg( \!\! - \! 2 \pi \!\!\! \int^{\infty}_{0} \! \! \! \zeta_{I}(x) \!\! \int^{\infty}_{x} \!\! \bigg[   1 \! - \! f_{\tau-\eta}   \Big(    
\frac{l^{ \alpha  }_{0}  l^{  \alpha \epsilon}  }{ x^{ \alpha}   l^{\alpha\epsilon}_{0} }     \Big)    \\
\hspace{37mm} \times  f_{\eta} \Big(  \frac{l^{ \alpha  }_{0}  l^{  \alpha \epsilon}   }{x^{ \alpha}   l^{\alpha\epsilon}_{0} }    \! \Big) \!    \bigg] \mathrm{d} l \mathrm{d} x  \! \bigg)  , \hspace{1mm} \mathrm{A}\! = \! \mathrm{Q}     \\ 
 \exp \! \bigg( \!\! - \! 2 \pi \!\! \int^{\infty}_{0} \! \! \! \zeta_{I}(x) \!\!\! \int^{\infty}_{x} \!\!  \bigg[   1 \! - \! f_{\tau-\eta}   \Big(      
\frac{l^{ \alpha  }_{0}  l^{  \alpha \epsilon}  }{x^{ \alpha}  l^{\alpha\epsilon}_{0} }   \! \Big)    \!    \bigg] \mathrm{d} l \mathrm{d} x  \! \bigg)    \\
\! \times \! \exp \! \bigg( \!\! - \! 2 \pi \!\! \int^{\infty}_{0} \! \!\!\! \zeta_{I}(x) \!\! \int^{\infty}_{x} \!\! \bigg[  1 \! - \!  f_{\eta} \Big(    
\frac{l^{ \alpha  }_{0}  l^{  \alpha \epsilon}   }{ x^{ \alpha}   l^{\alpha\epsilon}_{0} }   \! \Big)     \!    \bigg]\mathrm{d} l \mathrm{d} x  \! \bigg) \!   ,  \hspace{1mm} \mathrm{A} \!=\! \mathrm{F} \vspace{3mm}    
\end{dcases}   \nonumber \vspace{3mm}
\end{align}
where the interfering users for the initial transmission and the retransmission are averaged over the same inhomogeneous PPP and two independent PPPs for the cases with QSI and FVI, respectively.

Then, by 
inserting $f_{l_{0}}(r)$ into (\ref{eqn:C_3}) and changing the variable $t= C_{2} \pi \zeta_{B} r^2 $ and $B_{\!\widehat{P}}=C_{2}\pi\zeta_{B} (\frac{\widehat{P}}{\varrho} )^{ \frac{2}{\alpha \epsilon}}$,  the final results in (\ref{eqn:M_II}) can be obtained after some mathematical simplification.      
\end{proof}



      


\end{document}